\documentclass{article}
\pdfoutput=1
\usepackage{graphicx} % Required for inserting images
\usepackage{dcolumn}% Align table columns on decimal point
\usepackage{bm}% bold math
\usepackage{xcolor}
\usepackage{physics}
\usepackage{float}
\usepackage{amsfonts}
\usepackage{amsmath}
\usepackage{amsthm}
\usepackage{mathtools}
\usepackage{tcolorbox}
\usepackage{amssymb}
\usepackage{fullpage}
\usepackage{hyperref}
\usepackage{algorithm}
\usepackage{enumitem}
\usepackage[noend]{algpseudocode}

\definecolor{scmcolor}{RGB}{0,0,255}

\definecolor{scmcolorRed}{RGB}{255,0,0}

\newtheorem{theorem}{Theorem}
\newtheorem*{theorem*}{Theorem}
\newtheorem{definition}[theorem]{Definition}

\newtheorem*{corollary*}{Corollary}

\newtheorem{lemma}[theorem]{Lemma}

\newtheorem*{remark*}{Remark}

\title{
Improved separation between quantum and classical computers for sampling and functional tasks
}
\author{
Simon C. Marshall\thanks{Leiden University.}
\and
Scott Aaronson\thanks{University of Texas at Austin.}
\and
Vedran Dunjko\footnotemark[1]
}
\date{}

\begin{document}
\maketitle
\begin{abstract}
This paper furthers existing evidence that quantum computers are capable of computations beyond classical computers.
Specifically, we strengthen the collapse of the polynomial hierarchy to \textit{the second level} if: (i) Quantum computers with postselection are as powerful as classical computers with postselection ($\mathsf{PostBQP=PostBPP}$), 
(ii) any one of several quantum sampling experiments ($\mathsf{BosonSampling}$, $\mathsf{IQP}$, $\mathsf{DQC1}$) can be approximately performed by a classical computer (contingent on existing assumptions).
This last result implies that if any of these experiment's hardness conjectures hold, then quantum computers can implement functions classical computers cannot ($\mathsf{FBQP\neq FBPP}$) unless the polynomial hierarchy collapses to its 2nd level. These results are an improvement over previous work which either achieved a collapse to the third level or were concerned with exact sampling, a physically impractical case.

The workhorse of these results is a new technical complexity-theoretic result which we believe could have value beyond quantum computation. In particular, we prove that if there exists an equivalence between problems solvable with an exact counting oracle and problems solvable with an approximate counting oracle, then the polynomial hierarchy collapses to its second level, indeed to $\mathsf{ZPP^{NP}}$.
\end{abstract}

\maketitle
\section{Introduction}
Quantum computers are thought to be poised to outperform classical computers on a number of significant tasks, such as integer factorisation or simulation of quantum systems. Despite this widely held belief, we have no formal proof that factoring is not solvable in polynomial time on a classical computer, or even that $\mathsf{BQP\neq BPP}$.

This is perhaps not surprising given the challenge posed by proving unconditional results in complexity theory; famously $\mathsf{P\stackrel{?}{=}NP}$ remains unresolved after decades of study. A more realistic approach is to base the hardness of quantum computing on widely accepted complexity-theoretic conjectures.

This approach has so far been quite successful for sampling problems, with a number of papers \cite{aaronson2011computational, bremner2016average} showing that various quantum sampling experiments would be hard for a classical algorithm unless the polynomial hierarchy collapses to its third level, or else a widely believed conjecture fails. These results are especially interesting as they naturally imply that quantum computers could implement functions classical computers could not ($\mathsf{FBQP \neq FBPP}$) due to a surprising equivalence between sampling and functional classes \cite{aaronson2014equivalence}.

In a similar vein, there has been a body of work showing that various changes to quantum computing make quantum computing vastly more powerful than similarly modified classical computing. A well-known example is $\mathsf{PostBQP}$ vs $\mathsf{PostBPP}$, i.e. quantum/classical computing with postselection. Surprisingly, the former is in some sense equivalent to problems involving exact counting ($\mathsf{PP}$) \cite{aaronson2005postBQP}, while the latter is only equivalent to approximate counting \cite{o2018weakness}. 
This disconnect implies that the classes are likely not equal, indeed if they are this would also imply a collapse of the polynomial hierarchy to the third level \cite{aaronson2005postBQP}.

Both of these research lines provide strong evidence for that quantum computing is more powerful than classical computing, but they rely on a number of conjectures and complexity theory conditions.
We are therefore motivated to attempt to minimise or remove the need for these conjectures/conditions until we are left with compelling evidence that quantum computing is fundamentally stronger than classical computing.

Hope that weaker conditions might be possible was provided by Fujii et al. \cite{fujii2015power}, who showed that if classical computers can \textit{exactly sample} from any one of a number of quantum supremacy experiments then the polynomial hierarchy would collapse to its second level, an improvement over the existing collapse to third level. In fact they push the collapse all the way to $\mathsf{AM}$, in the second level of the hierarchy.
While this result provides insight as to what may be possible, it stops short of our goal as the proof strategy did not extend to additive approximate sampling, which is arguably the realistic case \cite{aaronson2011computational} and is relevant for questions such as $\mathsf{FBQP\stackrel{?}{=}FBPP}$.

In this paper, we provide the following new complexity theoretic result, which we then use to improve the abovementioned results to a 2nd level polynomial hierarchy collapse. In the following theorem, we use $\parallel$ to denote only one round of parallel oracle calls.
\begin{theorem*}[Main]
    If $\mathsf{P^{\#P}}=\mathsf{BPP_{\parallel}^{NP}}$ then $\mathsf{P^{\#P}=ZPP^{NP}}$
\end{theorem*}
By showing that the existing collapse results for classical boson sampling, commuting circuit sampling and 1 clean qubit sampling can all be modified to show $\mathsf{P^{\#P}}=\mathsf{BPP_{\parallel}^{NP}}$ (we are required to add the parallel requirement to $\mathsf{BPP^{NP}_{\parallel}}$ for our collapse), we improve the implied collapse of the polynomial hierarchy from 3rd order to 2nd. Consequently, this strengthens the evidence that $\mathsf{FBQP \neq FBPP}$. 
This theorem can also be made to strengthen the hierarchy collapse implied by assuming $\mathsf{PostBQP=PostBPP}$ from the third to the second level.

From a purely complexity-theoretic perspective, there is an interesting alternative form of this result in terms of exact and approximate counting\footnote{Exact counting can be defined as exactly counting the number of accepting inputs of some poly-time machine, approximate counting is getting within a multiplicative factor of 2 of exact counting (equivilently to within a $1+\frac{1}{poly(n)}$ factor).}.
\begin{center}
If $\mathsf{P^{ExactCount}_{\parallel}}=\mathsf{P^{ApproxCount}_{\parallel}}$ then $\mathsf{P^{\# P }=ZPP^{NP}}$
\end{center}
This naturally gives our main result the following interpretation: \textit{If the problems solved with the aid of exact counting can also be solved with approximate counting, then both are contained in $\mathsf{ZPP^{NP}}$}. 
The generality of this interpretation hints that there may be applications of this theorem to counting beyond the quantum theory considered herein.
It is also interesting to note that the proof contained herein does not relativise as we rely on the checkability of $\mathsf{\# P}$ for a step of our proof.

The paper is structured in 3 sections. Section \ref{sect:Approximate-exact counting theorem} focuses on the proof of our central theorem, section \ref{sect: postBQP v postBPP} proves that the polynomial hierarchy collapses to second order if $\mathsf{PostBQP=PostBPP}$. In the third section, we prove our various sampling results. Finally, we close by discussing interesting future research and open questions.

\section{Main theorem}\label{sect:Approximate-exact counting theorem}
In this section we prove our main theorem:

\begin{theorem}\label{thm: main}
    If $\mathsf{P^{\# P}=BPP_{\parallel}^{NP}}$ then $\mathsf{P^{\# P}=ZPP^{NP}}$.
\end{theorem}

We will assume readers are broadly familiar with classes such as $\mathsf{\# P}$, where classes are undefined we use the definitions given in the complexity zoo \cite{zoo}. We use the notation $\mathsf{A_{\parallel}^{B}}$ for $\mathsf{A}$ with one set of parallel oracle calls to $\mathsf{B}$.

It is useful to first give an informal description of the proof of theorem \ref{thm: main} before we proceed to the formal proof to give the reader a better idea of the purpose of each lemma. Broadly the proof relies on the notion of random-reducibility, where a given instance of a problem can be solved by a polynomial time algorithm with randomly selected instances of some other problem. If a problem can be randomly reduced to itself we say it is random self-reducible. 
The first key realisation for Theorem \ref{thm: main} is that a random self-reducible language that is in $\mathsf{BPP^{NP}_{\parallel}}$ is itself random reducible to $\mathsf{NP}$, because randomly selecting $x$ for a language in $\mathsf{BPP^{NP}_{\parallel}}$ gives rise to random set of non-adaptive $\mathsf{NP}$ calls, which meets the definition of random reducibility to $\mathsf{NP}$.
Next, as $\mathsf{\#P}$ is random self-reducible \cite{feigenbaum1990random}, the antecedent of Theorem \ref{thm: main} implies that it is random reducible to $\mathsf{NP}$.
The second key realisation is to notice that the proof by Feigenbaum and Fortnow that $\mathsf{NP}$ cannot be random self-reducible unless it is also in $\mathsf{coNP/poly}$ extends beyond self reducibility; any language that is random reducible to $\mathsf{NP}$ is in $\mathsf{coNP/poly}$. 
By repeating the argument for $\mathsf{coNP}$ we show $\mathsf{\#P\subseteq NP/poly \cap coNP/poly}$. Arvind and Köbler \cite{ARVIND2002257} showed any checkable language in $\mathsf{NP/poly \cap coNP/poly}$  is low for $\mathsf{\Sigma^P_2}$ which gives us a polynomial hierarchy collapse to $\mathsf{\Sigma_2^P}$. The final collapse to $\mathsf{ZPP^{NP}}$ comes from using the $\mathsf{BPP^{NP}_{\parallel}}$ algorithm to produce proofs verifiable in $\mathsf{P^{NP}}$, thereby achieving zero error.

We begin by defining random reducibility. A problem is randomly self-reducible if any given instance of the problem can be solved probabilistically by a polynomial time algorithm with access to solved random instances of the same problem.  
For our purposes we will work with a similar concept: random-reducibility, which is the same except the random instances may be from a different problem. It should be stated that this is different from a `randomized reduction', which is a randomised Karp reduction. Random reductions are sometimes called 1-locally random reductions to avoid this confusion.

\begin{definition}[Random Reducible]\label{def: random reducible}
    A function $f$ is \textbf{nonadaptively $k(n)$ random-reducible} to a function $g$ if there are polynomial time computable functions $\sigma$ and $\phi$ with the following two properties.
    \begin{enumerate}[label={(\arabic*)}]
        \item For all $n$ and all $x \in \{0,1\}^n$ $f(x)$ can probably be solved by $\phi$ with instance of $g(y)$ for $y$ selected by $\sigma$:
            \[
            \Pr_{r}\big(
            f(x)=\phi(x,r,g(\sigma(1,x,r)),\ldots, g(\sigma(k, x, r)))
            \big)\geq2/3
            \]
        \item For all $n$, all pairs $\{x_1,x_2\}\subset \{0,1\}^n$ and all $i$, if $r$ is chosen uniformly at random then $\sigma(i,x_1,r)$ and $\sigma(i,x_2,r)$ are identically distrubuted.
    \end{enumerate}
    If both conditions hold we say ($\sigma$, $\phi$) is a random-reduction from $f$ to $g$.
\end{definition}

As we are dealing with no other forms of random reducibility we will just say `$f$ is rr to $g$' when $f$ is non-adaptively $k(n)$ random-reducible to $g$ for some polynomial $k$. If a function is rr to itself then we say the function is random self-reducible (rsr). A set is rr to another if its characteristic function is rr. A set of languages, $\mathsf{A}$ is rr to $\mathsf{B}$ if every $L\in \mathsf{A}$ is rr to some $L'\in\mathsf{B}$.

As with probabilistic classes like $\mathsf{BPP}$, we can boost random reducibility an arbitrarily low ($2^{-n}$) failure probability.

\begin{lemma}[\cite{feigenbaum1990random}]\label{lma: boost rr}
    If function $f$ is non-adaptively $k(n)$ random-reducible to $g$ then for all $t(n)$, $f$ is non-adaptively $24t(n)k(n)$ random-reducible to $g$ where condition (1) holds for at least $1-2^{-t(n)}$ of the r's in $\{0,1\}^{24t(n)k(n)}$.
\end{lemma}

\begin{remark*}
    In general, this boosting may not work for a definition of random reducibility dealing with relations instead of functions as they may have multiple correct outputs. However, counting problems like $\mathsf{Perm}$, which has only one correct output for a given input, can be boosted.
\end{remark*}

\begin{lemma}[\cite{feigenbaum1990random}]\label{lma: perm is rsr}
    Any $\mathsf{\#P}$-complete language is rsr.
\end{lemma}

This concludes our definitions, we can now state an intermediate theorem to our final result.

\begin{theorem}\label{thm: BPP is rr to NP}
    If $\mathsf{P^{\# P}= BPP_{\parallel}^{NP}}$ then $\mathsf{P^{\# P}}$ is random reducible to $\mathsf{NP}$.
\end{theorem}

Before we prove this, we will provide a number of smaller results. In the following lemma the notation $\mathsf{A^{(B[1])}}$ means $\mathsf{A}$ with one oracle call to $\mathsf{B}$.

\begin{lemma}\label{lma: P to 1 sharp P is rr to 1 sharp P}
    Any language in $\mathsf{P^{\# P[1]}}$ is random reducible to $\mathsf{\# P}$.
\end{lemma}
\begin{proof}
    As $\mathsf{\# P}$-complete languages are random self-reducible and only one oracle call is needed we can use a random reduction of $\mathsf{\# P}$ to give the answer to the oracle call and use the polynomial time $\phi$ algorithm to perform the rest of the $\mathsf{P}$ algorithm.
\end{proof}

The next lemma captures the intuition that $\mathsf{BPP^{NP}_{\parallel}}$ algorithms are `almost' rr to $\mathsf{NP}$, only missing the random element selection part of the definition (property (2)).

\begin{lemma}\label{lma:BPP rr similarity}
    For all $L$ in $\mathsf{BPP^{NP}_{\parallel}}$ there exists polynomial time computable functions $\sigma_{B}$, $\phi_{B}$ such that
    \[
    \Pr_{r}\left(
            L(x)=\phi_{B}(x,r,\mathsf{SAT}(\sigma_{B}(1,x,r)),\ldots, \mathsf{SAT}(\sigma_{B}(m, x, r)))
            \right)\geq 1-2^{-n}.
    \]
\end{lemma}
\begin{proof}
    If $L\in \mathsf{BPP^{NP}_{\parallel}}$ then there exists a polynomial time algorithm, $A$, to decide $x \in L$ which calls only one set of up to $m$ non-adaptive $\mathsf{NP}$ queries. We assume the $\mathsf{NP}$ calls $\mathsf{SAT}$ for simplicity. Let the algorithm for $\sigma_{B}(i,x,r)$ run $A$ until the step involving oracle queries and then output just the $i$'th query (assuming some arbitrary query ordering).
    The algorithm for $\phi_{B}$ is now just the algorithm $A$ using the oracle calls produced asked by $\sigma_{B}(i,x,r)$ in place of directly making its own calls. Since both $\phi_{B}$ and $\mathsf{\sigma_{B}}$ have access to the same random string $r$ they will both select the same oracle calls.
\end{proof}
Lemma \ref{lma:BPP rr similarity} is useful as it proves that $\mathsf{BPP^{NP}_{\parallel}}$ fulfils property (1) of the definition of random reducibility to $\mathsf{NP}$, but does not prove the random distribution of ${\sigma_B(1,x,r)}$. The next theorem shows that if a langauge already randomly reduces to $\mathsf{BPP^{NP}_{\parallel}}$ then this random reduction condition (condition (2)) is also fufilled, and thus the language is random reducible to $\mathsf{NP}$

\begin{lemma}\label{lma: rr to BPP NP parallel is rr to NP}
    All languages random reducible to a language in $\mathsf{BPP^{NP}_{\parallel}}$ are random reducible to $\mathsf{NP}$.
\end{lemma}
\begin{proof}Let $L$ be a language which is random reducible to $L_B\in \mathsf{BPP^{NP}_{\parallel}}$. We will demonstrate that $L$ is rr to $\mathsf{NP}$ by writing out the definition of random reducibility to ${L_B}$ then using lemma \ref{lma:BPP rr similarity} to substitute the calls to ${L_B}$ with a formula using calls to $\mathsf{SAT}$. We can then rearange this into a format which directly makes calls to $\mathsf{SAT}$ and we show that these calls inherit the randomness property from the random reduction of $L$ to $L_B$.

    Let us assume the random reduction of $L$ to $L_B$ on input $x$ uses $k$ calls, from the definition of random reducibility there exists $\phi$ and $\sigma$ such that
    \[
\Pr_{r}\left(
            L(x)=\phi(x,r,L_{\mathsf{BPP}}(\sigma(1,x,r)),\ldots, L_{{B}}(\sigma(k, x, r)))
            \right)\geq 1-2^{-n}.
\]

Define $y_i:=\sigma(i, x, r)$. 

\[
\Pr_{r}\left(
            L(x)=\phi(x,r,L_{\mathsf{B}}(y_1),\ldots, L_{\mathsf{BPP}}(y_k))
            \right)\geq 1-2^{-n}
\]

Using lemma \ref{lma:BPP rr similarity} we can substitute $L_B$ with $\phi_B$, $\sigma_B$ and a random string $r_i$ for the i'th call to $L_B$. In the following we assume each lemma-\ref{lma:BPP rr similarity}-reduction uses $m$ $\mathsf{SAT}$ calls.

\begin{align*}
\Pr_{r, r_1, \ldots r_k}\big(
            L(x)=\phi(x,r,&\phi_{B}(y_1,r_1,\mathsf{SAT}(\sigma_{B}(1,y_1,r_1)),
            \ldots, \mathsf{SAT}(\sigma_{B}(m, y_1, r_1))),\\ \ldots, &\phi_{B}(y_k,r_k,\mathsf{SAT}(\sigma_{B}(1,y_k,r_k)), \ldots, \mathsf{SAT}(\sigma_{B}(m,y_k,r_k)))
            \big)
            \geq 1-(k+1)2^{-n}.
\end{align*}

The failure probability $1-(k+1)2^{-n}$ comes the failure probability of the $k$ reductions combined with the failure probability from the random reduction.

We can now begin to combine elements of the reduction from $L$ to $L_B$ with reduction from $L_B$ to $\mathsf{SAT}$. We start with $\phi$ which we combine with $\phi_B$ to define $\phi'$, informally we can think of this as doing the two polynomial-time algorithms as a larger polynomial time algorithm.
\begin{multline*}
\phi'(x,r,r_1,\ldots,r_k, \mathsf{SAT(\sigma_B(1,y_1,r_1})),\ldots \mathsf{SAT(\sigma_B(m,y_k,r_k}))):=
\\
\phi(x,r,\phi_{B}(y_1,r_1,\mathsf{SAT}(\sigma_{B}(1,y_1,r_1)),
\ldots, \mathsf{SAT}(\sigma_{B}(m, y_1, r_1))),\ldots,\phi_{B}(y_k,r_k,\mathsf{SAT}(\sigma_{B}(1,y_k,r_k)), \ldots, \mathsf{SAT}(\sigma_{B}(m,y_k,r_k))
\end{multline*}

We use `$x\#y$' to denote $y$ appended to $x$. Define $\mathbf{r}:=r\#r_1\#\ldots\#r_k$.

Again we define a new function, $\sigma'$, as the combination of two existing function, $\sigma$ and $\sigma_B$:

\[
\sigma'(i,j,x,\mathbf{r}):= \sigma_B(i,\sigma(j,x,r),r_j).
\]

Which gives the following simplified equation:
\begin{equation}\label{eq: finally}
\Pr_{\mathbf{r}}\big(
            L(x)=
\phi'(x,\mathbf{r}, \mathsf{SAT}(\sigma'(1,1,x,\mathbf{r})),\ldots \mathsf{SAT}(\sigma'(m,k,x,\mathbf{r}))
            \big)
            \geq 1-(k+1)2^{-n}.
\end{equation}

We will now directly check the definition of $\sigma'$, $\phi'$ and the equation \ref{eq: finally} against the definition of random-reducbile, recall the two conditions definition \ref{def: random reducible}:
\begin{enumerate}[label={(\arabic*)}]
        \item For all $n$ and all $x \in \{0,1\}^n$ $f(x)$ can probably be solved by $\phi$ with instance of $g(y)$ for $y$ selected by $\sigma$:
            \[
            \Pr_{r}\big(
            f(x)=\phi(x,r,g(\sigma(1,x,r)),\ldots, g(\sigma(k, x, r)))
            \big)\geq2/3
            \]
        \item For all $n$, all pairs $\{x_1,x_2\}\subset \{0,1\}^n$ and all $i$, if $r$ is chosen uniformly at random then $\sigma(i,x_1,r)$ and $\sigma(i,x_2,r)$ are identically distrubuted.
    \end{enumerate}

We can see that these two conditions are met: 

\begin{enumerate}[label={(\arabic*)}]
        \item For sufficiently large $n$, $1-(k+1)2^{-n}>2/3$ therefore equation \ref{eq: finally} is of the form:
            \[
            \Pr_{r}\big(
            f(x)=\phi(x,r,\mathsf{SAT}(\sigma(1,x,r)),\ldots, \mathsf{SAT}(\sigma(k, x, r)))
            \big)\geq2/3.
            \]
        \item For all $n$, all pairs $\{x_1,x_2\}\subset \{0,1\}^n$, all $i$ and $j$, if $r$ is chosen uniformly at random then $\sigma(j,x_1,r)$ and $\sigma(j,x_2,r)$ are identically distributed (as per their definition), therefore $\sigma'(i\#j,x_{1/2},\mathbf{r})$ is identically distributed too. This is because the only dependence on $x$ is through the identically distributed $\sigma$: $\sigma'(i\#j,x,\mathbf{r})= \sigma_\mathsf{B}(i, \sigma(j, x, r), r_j)$. This holds for all $r_i$, $r_j$. 
    \end{enumerate}
Thus fufilling the conditions for $L$ to be random reducible to $\mathsf{NP}$
\end{proof}

Finally, we can proceed with the proof of theorem \ref{thm: BPP is rr to NP}

\begin{proof}[Proof of theorem \ref{thm: BPP is rr to NP}]
Toda \cite{toda1991pp} showed that $\mathsf{PH\subseteq P^{\# P[1]}}$. Therefore, by the assumption of the theorem if $\mathsf{P^{\# P}= BPP_{\parallel}^{NP}=PH}$ then $\mathsf{P^{\# P} = P^{\# P[1]}}$. 

By lemma \ref{lma: P to 1 sharp P is rr to 1 sharp P} we know $\mathsf{P^{\# P}}$ is rr to $\mathsf{\# P}$. Under the assumption of this theorem, $\mathsf{\# P \subseteq BPP^{NP}_{\parallel}}$, $\mathsf{P^{\# P}}$ is rr to $\mathsf{BPP^{NP}_{\parallel}}$. By lemma \ref{lma: rr to BPP NP parallel is rr to NP} this means $\mathsf{P^{\# P}}$ is rr to $\mathsf{NP}$.
\end{proof}

The next result is heavily based on the work of Feigenbaum and Fortnow \cite{feigenbaum1990random}. They show that if $\mathsf{NP}$ is random self reducible then $\mathsf{coNP}$ is in $\mathsf{NP/poly}$, while the result we are trying to show is slightly different than this, the method is very similar.

\begin{theorem}\label{thm: rr to NP is NP/poly}
    Any language that is random-reducible to a language in $\mathsf{NP}$ is also in $\mathsf{coNP/poly \cap NP/poly}$.
\end{theorem}

\begin{proof}
    Let $\mathsf{AM^{poly}}$ be $\mathsf{AM}$ with polynomial advice given to Arthur\footnote{ This is not the same as $\mathsf{AM/poly}$, as the latter must be an $\mathsf{AM}$ language for all advice, whereas the former only needs to be defined for the correct advice.}. It turns out that $\mathsf{AM^{poly}=NP/poly}$ and $\mathsf{coAM^{poly}=coNP/poly}$ \cite{feigenbaum1990random}.
    We will prove $L$ is in $\mathsf{AM^{poly}\cap coAM^{poly}}$, beginning with $L \in \mathsf{AM^{poly}}$.

    Suppose $L(x)$ is non-adaptively $k(n)$ random-reducible to $\mathsf{SAT}(x)$ with error probability $2^{-n}$. We will adopt the notation $y_i=\sigma(i,x,r)$. For instance size $n$ we will give Arthur the advice $(p_1,\ldots, p_k)$ where $p_i$ is the probability that $y_i=1$,
    \[
    p_i=\Pr_{r}(\mathsf{SAT}(\sigma(i, x,r))=1).
    \]
    As $\sigma(i, x,r)$ is distributed identically (given uniform random $r$) for all $x$ this advice does not depend on the input $x$.

    The proof system will consist of Arthur selecting $m=9k^3$ random strings, $\{r_j\}_{j\in[m]}$, and passing this to Merlin, these random strings defines which $\mathsf{SAT}$ queries, $y_{0,j},\ldots, y_{k,j}$, will be made. Merlin will hand back $m$ `transcripts'. These transcripts consist of a list $w_{i,j}$ which is either a witness for $y_{i,j}$ or is the string `NIL' (which is interpreted as the claim that $y_{i,j}\notin \mathsf{SAT}$). For each of the $m$ random strings we get the transcript:
    \[
    Transcript(x,r_j)= (w_{0,j},w_{1,j}, \ldots, w_{k,j})
    \]
    Arthur performs three checks:
    \begin{enumerate}[label={(\arabic*)}]
        \item For all $i,j$ either the witness $w_{i,j}$ is `NIL' or he checks the witness is valid for $y_{i,j}\in \mathsf{SAT}$
        \item Let $b_{i,j}$ be 0 if $w_{i,j}$ is `NIL' and one otherwise. For all $j\in [m]$ Arthur checks that $\phi(x,r_j, b_{0,j}, \ldots, b_{k,j})=1$, i.e. If Merlin has told the truth then $\phi$ will accept.
        \item For each $i\in [k]$ Arthur checks that more than $p_im-2\sqrt{km}$ of the $y_{i,j}$ have been proved to be in $\mathsf{SAT}$, if this condition does not hold he rejects.
    \end{enumerate}
    If these checks all pass, then Arthur accepts. 
    The trick in this proof is this final condition. Over the random strings, each $y_i$ has some probability of being in $\mathsf{SAT}$ with a given variance, by picking $m$ to be large we force this variance to be small. With high probability, a correct answer lies in this range. Since Merlin cannot lie about yes instances (since he must prove them), he can only lie about no instances. However, if he lied too much it would be clear that not enough of $y_i$ are in $\mathsf{SAT}$, thus we could probabilistically guess he was cheating. To exploit this, $m$ is picked precisely so Merlin cannot cheat on all $j$ without exceeding his `lying budget'. We will now formally show soundness and completeness.

    \textit{Completness.} If $x\in L$ and Merlin is honest, providing witnesses for all $y_{i,j}$ in $\mathsf{SAT}$, condition (1) will always hold. Condition (2) holds with probability at least $1-2^{-n}$ (by the definition of rr). Therefore we just need to show condition (3) holds with high probability.

    Let $Z_{i,j}:=(y_{i,j}=1)$ and $Z_{i}=\sum_{j=1}^m Z_{i,j}$. We defined our advice so $p_i= \mathbb{E}[Z_i]/m$, we can derive that $\operatorname{Var}(X)=p_i(1-p_i)m<m$. We can use these properties and Chebyshev's inequality to bound our probability of failing check (2).
    \begin{align*}
        \Pr(Z_i\leq p_im-2\sqrt{km})\\
        \leq\Pr(\norm{Z_i-p_im}\geq 2\sqrt{km})\\
        =\Pr(\norm{Z_i-\Bar{Z_i}}\geq 2\sqrt{km})\\
        \leq \operatorname{Var}(X)/4km\leq 1/(4k)\leq1/4
    \end{align*}
    The probability of failing any check given $x\in L$ is less than $\frac{1}{4}+2^{-n}$ which is less than $1/3$ for large enough $n$. This proves completeness

    \textit{Soundness.} Suppose $x \notin L$, if Arthur accepts then all 3 checks must have passed for all $j\in[m]$. If some $y_{i,j}$ is in $\mathsf{SAT}$ but Merlin has provided $w_{i,j}=`NIL'$ then we say Merlin has lied about $y_{i,j}$, if check (1) passes Merlin has not lied about any `yes' instances, so we disregard this case.
    The probability of the random reduction failing is $2^{-n}$ without any lies, so for all $m$ reductions to fail Merlin must lie $m$ times with probability $(1-2^{-n})^m> 1-m2^{-n}$.

    If Merlin lies $m$ times there must be an $i$ that he claims at least $m/k$ of the $y_{i,j}$ are not in $\mathsf{SAT}$ when they are. To satisfy condition (3) at least $p_im-2\sqrt{km}\mathbf{+m/k}$ of the $y_{i,j}$ must be in $\mathsf{SAT}$ to leave `room' for Merlin to lie $m/k$ times without violating (3). The probability of this is given by a Chernoff bound on the random variable $Z_i$, as defined above.
    \begin{align*}
    \Pr\left(Z_i\geq p_im-2\sqrt{km}+m/k\right)\\
    =\Pr\left(Z_i-p_im\geq m/k-2\sqrt{km}\right)\\
    =\Pr\left(Z_i-p_im\geq 9k^3/k-2\sqrt{k9k^3}\right)\\
    =\Pr\left(Z_i-p_im\geq 3k^2\right)\\
    \leq e^{-2\times 9k^4/9k^3}=e^{-2\times k}\\
    \leq 1/(4k)
    \end{align*}
    Combining this with the normal probability that the random reduction fails even with correct oracle answers (which had a $2^{-n}$ probability) gives an acceptance probability given $x\notin L$ of less than $m2^{-n}+1/(4k)<1/3$ for large enough $n$. This proves soundness.
    
    The previous analysis can just as easily be repeated for $\mathsf{coAM^{poly}}$ with Merlin proving no instances, giving $L\in \mathsf{AM^{poly} \cap coAM^{poly}}$. By Feigenbaum and Fortnow we know this equals $\mathsf{NP/poly \cap coNP/poly}$.
\end{proof}

To prove the next Theorem we must recall a crucial result by Arvind and Köbler, which checkability \cite{ARVIND2002257,zoo}. The notion of checkability is somewhat involved but intuitively has to do with whether efficent programs that decide the set can themselves be efficently checked. Since we will just need the checkability of $\mathsf{PP}$ and $\mathsf{\# P}$ to deploy the following theorem we will not formally define it, instead we refer the reader to \cite{ARVIND2002257, blum1995designing}.

\begin{theorem*}[\cite{ARVIND2002257}, Theorem 22]
    Checkable sets in $\mathsf{NP/poly \cap coNP/poly}$ are low for $\mathsf{\Sigma_2^P}$.
\end{theorem*}

This result is imperative to proceed, as while it is true that $\mathsf{NP\subset (NP\cap coNP)/poly}$ implies $\mathsf{PH=ZPP^{NP}}$, the same is not true for $\mathsf{NP\subset (NP/poly)\cap (coNP/poly)}$.

\begin{lemma}\label{lma: sharp P in conp/poly}
    If $\mathsf{P^{\# P}}\subset (\mathsf{coNP/poly \cap NP/poly})$ then $\mathsf{P^{\# P}=\Sigma_2^P\cap \Pi_2^P}$
\end{lemma}
\begin{proof}
By Toda $\mathsf{P^{\#P}=P^{PP}}$ \cite{toda1991pp}.    
    If $\mathsf{P^{PP}}$ is in $\mathsf{\mathsf{coNP/poly \cap NP/poly}}$ then clearly so is $\mathsf{PP}$, combining this with the existence of $\mathsf{PP}$-complete checkable sets \cite{ARVIND2002257} we know that $\mathsf{PP}$ is low for $\mathsf{\Sigma_2^P}$.

    This lowness implies we can put  $\mathsf{P^{\# P}}$ in $\mathsf{\Sigma_2^P}$ by the following series of inclusions
    \[
    \mathsf{
    P^{\#P}=P^{PP}\subseteq (\Sigma_2^P)^{PP} = \Sigma_2^P.
    }
    \]

    As $\mathsf{\Sigma_2^P\subseteq P^{\# P}}$ by Toda \cite{toda1991pp} we can fix the previous inclusion into an equality: $\mathsf{\Sigma_2^P =  P^{\# P}}$.

    As $\mathsf{\Pi_2^P \subseteq P^{\# P}}$ (Toda \cite{toda1991pp}) we get $\mathsf{\Pi_2^P\subseteq \Sigma_2^P}$.
    Which gives us the final equality:
    \[
    \mathsf{P^{\# P} \subseteq \Sigma_2^P \subseteq P^{\Pi_2^P}= P^{\Sigma_2^P\cap\Pi_2^P}=\Sigma_2^P\cap\Pi_2^P}
    \]
\end{proof}

The previous lemma has achieved a collapse to the second level but we can collapse to $\mathsf{ZPP^{NP}}$ by using the proveability of languages in $\mathsf{\Sigma_2^P\cap \Pi_2^P}$ to make the $\mathsf{BPP^{NP}_{\parallel}}$ algorithm zero-error.

\begin{lemma}\label{lma: sharp P in ZPP}
    All languages in $\mathsf{BPP^{NP}\cap\Sigma_2^P\cap \Pi_2^P}$ are in $\mathsf{ZPP^{NP}}$
\end{lemma}
\begin{proof}
    Fix some $L\in \mathsf{BPP^{NP}\cap\Sigma_2^P\cap \Pi_2^P}$. We can check if $x\in L$ using the $\mathsf{BPP^{NP}}$ algorithm. Use the $\mathsf{BPP^{NP}}$ program to determine a proof of this fact for either the $\mathsf{\Sigma_2^P}$ verifier or the $\mathsf{\Pi_2^P}$ verifier. We can test either proof in $\mathsf{P^{NP}}$ using the verifiers from either $\mathsf{\Sigma_2^{P}}$ or $\mathsf{\Pi_2^{P}}$. With probability $1-2^{-n}$ the $\mathsf{BPP^{NP}}$ algorithm succeeds in producing a valid proof in polynomial time, any valid proof proves or disproves $x\in L$. Therefore in polynomial time the algorithm has successfully determined if $x\in L$ with probability $1-2^{-n}$ and never incorrectly answers, placing $L\in \mathsf{ZPP^{NP}}$.
\end{proof}

This completes the proof of our main theorem which is given below.

\begin{theorem*}
    If $\mathsf{P^{\# P} = BPP^{NP}_{\parallel}}$ then $\mathsf{P^{\# P} = ZPP^{NP}}$
\end{theorem*}
\begin{proof}[Summary of proof]
    If $\mathsf{P^{\# P} = BPP^{NP}_{\parallel}}$ then $\mathsf{P^{\#P}}$ is random reducible to $\mathsf{NP}$ (Theorem \ref{lma: perm is rsr}).

    If $\mathsf{P^{\# P}}$ is random reducible to $\mathsf{NP}$ then $\mathsf{P^{\#P}}$ is in $\mathsf{coNP/poly \cap NP/poly}$ (Theorem \ref{thm: rr to NP is NP/poly}).

    If $\mathsf{P^{\#P}\subseteq coNP/poly \cap NP/poly}$ then $\mathsf{P^{\#P}} = \mathsf{\Sigma_2^P\cap \Pi_2^P}$ (Lemma \ref{lma: sharp P in conp/poly}).

    If $\mathsf{P^{\#P}} = \mathsf{\Sigma_2^P\cap \Pi_2^P}$ and $\mathsf{P^{\#P}= BPP^{NP}}$ then $\mathsf{P^{\#P}= ZPP^{NP}}$ (Lemma \ref{lma: sharp P in ZPP}).
\end{proof}

\begin{remark*}
    Interestingly we can perform all of the oracle calls required to produce the proof of $x\in L$ in one parallel step, it then only takes one extra Oracle call to check this proof. Thus it is possible to answer correctly with probability $1-2^{-n}$ or with `I don't know' after only 2 rounds of parallel $\mathsf{NP}$ oracle calls.
\end{remark*}

\begin{remark*}
    Feigenbaum and Fortnow \cite{feigenbaum1990random} showed their proof applies beyond non-adaptive random self-reducibility, up to logarithmically many adaptive steps are allowed. For the same reasons, our proof could be extended to logarithmically many rounds of polynomially many parallel oracle calls.
\end{remark*}

The following theorem formalises the `natural interpretation' of our result given in the intro.

\begin{theorem}[Natural interpretation]\label{thm: natural}
    If $\mathsf{P^{ExactCount}_{\parallel}}=\mathsf{P^{ApproxCount}_{\parallel}}$ then $\mathsf{P^{\#P}=ZPP^{NP}}$
\end{theorem}

This result becomes easy to show after Theorem \ref{thm: postbqp postbpp}, so we will prove it at the end of the next section.

%%%%%%%%%%%%%%%%%%%%%%%%%%
%%%%%%%%  postBQP %%%%%%%%
%%%%%%%%%%%%%%%%%%%%%%%%%%
\section{PostBQP and PostBPP}\label{sect: postBQP v postBPP}

An immediate consequence of the last section's main theorem is to the question of $\mathsf{PostBQP\stackrel{?}{=} PostBPP}$, i.e. is a poly-time quantum computer with postselection equal to a classical computer with postselection? The best existing answer was given when Aaronson showed $\mathsf{PostBQP=PP}$ \cite{aaronson2014equivalence}, thus equality would imply $\mathsf{PH=BPP^{NP}}$ (a collapse to the third level). However, as $\mathsf{PostBPP}\subseteq \mathsf{BPP^{NP}_{\parallel}}$ \cite{o2018weakness, han1997threshold} our theorem can be used to improve this collapse to the second level. Before we provide proof of this we formally define $\mathsf{PostBQP}$ and the operator $\mathsf{BP\cdot}$ which makes a complexity class probabilistic.

\begin{definition}
$\mathsf{PostBQP}$ is the set of languages for which there exists a uniform family of polynomial width and polynomial depth quantum circuits $\{C_n\}_ {n\geq 1}$ such that for any input x,
\begin{enumerate}[label=(\roman*)]
    \item There is a non-zero probability of measuring the first qubit of $C_n \ket{0\ldots 0}\ket{x}$ to be $\ket{1}$.
    \item If $x\in L$, conditioned on the first qubit being $\ket{1}$, the second qubit is $\ket{1}$ with probability at least 2/3.
    \item If $x\notin L$, conditioned on the first qubit being $\ket{1}$, the second qubit is $\ket{1}$ with probability at most 1/3.
\end{enumerate}
\end{definition}

$\mathsf{PostBPP}$ can be defined similarly with classical circuits and additional input of random bits.

\begin{definition}
    Let K be a complexity class. A language $L$, is in $\mathsf{BP\cdot K}$ if there exists a language $\mathsf{A \in K}$ and a polynomial $p$ such that for all strings $x$.
$$Pr(r \in \{0,1\}^{p(|x|)}: x \in L \iff x\#r \in A) \geq 2/3$$
\end{definition}

This covers the necessary definitions and we can proceed with the collapse result of interest.

\begin{theorem}\label{thm: postbqp postbpp}
    If $\mathsf{PostBQP=PostBPP}$ then $\mathsf{P^{\#P}=ZPP^{NP}}$ and the polynomial hierarchy collapses to the second level.
\end{theorem}

\begin{proof}
Assume $\mathsf{PostBQP=PostBPP}$. By Toda $\mathsf{PH\subseteq P^{\#P}}$, by Aaronson \cite{aaronson2005postBQP} $\mathsf{PostBQP=PP}$ and by $\mathsf{PostBPP \subseteq BPP^{NP}_{\parallel}}$ \cite{o2018weakness, han1997threshold}. This gives: 
\[
\mathsf{
PP \subseteq BPP^{NP}_{\parallel} \text{  and  }  P^{\#P} = P^{PP} \subseteq P^{BPP^{NP}} = BPP^{NP} = PH \subseteq P^{\#P}.
}
\]

Therefore $\mathsf{P^{\#P}=PH}$.

At this point we recall an interesting lemma from Toda \cite{toda1991pp} to continue.
\[\mathsf{
PP^{PH}\subseteq BP\cdot PP
}\]

Clearly $\mathsf{PH\subseteq PP^{PH}}$ and it can be shown that if $\mathsf{PP\subseteq BPP^{NP}_{\parallel}}$ then $\mathsf{BP \cdot PP \subseteq BPP^{NP}_{\parallel}}$ (a full proof follows in lemma \ref{lma: BP PP is in parallel too}), therefore

\[\mathsf{
PH\subseteq BP\cdot PP \subseteq BPP^{NP}_{\parallel}
}.\]

Since $\mathsf{PH = P^{\#P}}$ and $\mathsf{BPP^{NP}_{\parallel}\subseteq PH}$ we can conclude $\mathsf{P^{\#P}=BPP^{NP}_{\parallel}}$. We now have the condition $\mathsf{P^{\#P}= BPP^{NP}_{\parallel}}$ so by theorem \ref{thm: main} we get $\mathsf{P^{\#P} = ZPP^{NP}}$, the final result.
\end{proof}

\begin{lemma}\label{lma: BP PP is in parallel too}
    If $\mathsf{PP}$ is in $\mathsf{BPP^{NP}_{\parallel}}$ then so is $\mathsf{BP \cdot PP}$
\end{lemma}
\begin{proof}
Fix some $L\in \mathsf{BP\cdot PP}$. By the definition, there exists $A\in\mathsf{PP}$ which defines $L$.
Assume $\mathsf{PP}$ is in $\mathsf{BPP^{NP}_{\parallel}}$, therefore $A$ is in $\mathsf{BPP^{NP}_{\parallel}}$. Therefore $L$ is in $\mathsf{BP\cdot BPP^{NP}_{\parallel}}$.

Note that majority-vote amplification can be used on $\mathsf{BP\cdot PP}$ and $\mathsf{BPP^{NP}_{\parallel}}$ to decrease the probability of failure to less than $1/6$.

If we combine the random coin flips of both the $\mathsf{BP \cdot PP}$ and $\mathsf{BPP}$ parts of the algorithm we get a single failure probability below 1/3, which fits the definition of $\mathsf{BPP^{NP}_{\parallel}}$.
\end{proof}

Theorem \ref{thm: postbqp postbpp} provides a quick a simple proof of Theorem \ref{thm: natural} (the natural interpretation of our result).

\begin{proof}[Proof of Theorem \ref{thm: natural}]
By definition $\mathsf{\#P=ExactCount}$. O'Donnell and Say show $\mathsf{P^{ApproxCount}_{\parallel}=PostBPP}$ \cite{o2018weakness}. 
Therefore we can rewrite the antecedent of the Theorem as $\mathsf{P^{\#P}_\parallel=PostBPP}$.

Any problem in $\mathsf{PP}$ can be solved with one $\mathsf{\#P}$ query, which can obviously be done in one parallel round of oracle calls, so $\mathsf{PP\subseteq P^{\#P}_\parallel}$. This gives $\mathsf{PP=PostBPP}$ (as we already know $\mathsf{PostBPP\subseteq PP}$ and by Theorem \ref{thm: postbqp postbpp} $\mathsf{P^{\#P}=ZPP^{NP}}$.
\end{proof}

%%%%%%%%%%%%%%%%%%%%%%%%%%
% Approximate boson Samp %
%%%%%%%%%%%%%%%%%%%%%%%%%%

\section{Improved Hardness of $\mathsf{BosonSampling}$, $\mathsf{IQP}$, and $\mathsf{DQC1}$}\label{sect: quantum supremacy}

In the original work on the hardness of Boson Sampling \cite{aaronson2011computational} two cases were considered: \textit{the exact case}, where the probability that the algorithm would sample a given element must be at least multiplicatively close to the target distribution, and \textit{the approximate case}, which allowed additive error in the total variation distance between the sampled and target distribution. 
For the exact case, classical simulation implied the polynomial hierarchy collapsed to the third level. For the approximate case, a collapse to the 3rd level was achieved, but subject to additional assumptions (we refer to these as additional assumptions henceforth). This separation between multiplicative and additive error also applies to the work on instantaneous quantum polynomial-time sampling ($\mathsf{IQP}$) \cite{bremner2016average} and sampling from 1 clean qubit circuits ($\mathsf{DQC1}$) \cite{morimae2017hardness}, albeit with different additional assumptions.
It is important to notice that realistic quantum computers are believed not to be able to achieve multiplicative error in sampling, but can achieve additive error \cite{bremner2016average, aharonov2023polynomial}. So it is the approximate case that distinguishes quantum from classical computation.

Fujii \cite{fujii2015power} et al strengthened the collapse for the exact case to $\mathsf{PH = AM \cap coAM}$ (a collapse to the second level of the polynomial hierarchy) for $\mathsf{BosonSampling}$, $\mathsf{DQC1}$ and $\mathsf{IQP}$ (amongst others). However, for fundamental reasons, the proof did not extend to the approximate case.

In this section, we show that efficient classical sampling of the physically relevant \textit{approximate case} would also collapse the polynomial hierarchy to its second level ($\mathsf{ZPP^{NP}}$), conditional on the existing additional assumptions. In this sense, our work provides the strongest known separation between classical and quantum computations for sampling and functional problems. While our result can also be applied to the exact case it is an not improvement on Fujii's result, as it is known $\mathsf{AM\cap coAM \subseteq ZPP^{NP}}$, and this application is therefore not relevant.

To apply Theorem \ref{thm: main} and show a second level $\mathsf{PH}$ collapse we require the condition $\mathsf{P^{\#P}=BPP^{NP}_{\parallel}}$, however, the proofs contained in \cite{aaronson2011computational, bremner2016average, morimae2017hardness} only show $\mathsf{P^{\#P}=BPP^{NP}}$. 
Fortunately, each of the proofs can be easily modified to only use parallel oracle calls, giving $\mathsf{P^{\#P}=BPP^{NP}_{\parallel}}$. 
Formally proving this fact would require reproducing each proof in detail and would make this paper excessively long.
Instead of completely rewriting these proofs we will instead notice that each proof only uses the $\mathsf{NP}$ oracle to approximately count some post-selected quantity with Stockmeyers algorithm \cite{stockmeyer1985approximation}. 
In the next lemma, we show that this computation can be done with parallel oracle calls. This will allow us to argue classical sampling of $\mathsf{BosonSampling}$, $\mathsf{IQP}$ or $\mathsf{DQC1}$ would imply $\mathsf{P^{\#P}=BPP^{NP}_{\parallel}}$ without formally reproducing each of the proofs.

\begin{lemma}[Counting a postselected quantity]\label{lma: postselected stockmeyer}
    Given a Boolean function $f:\{0,1\}^n\rightarrow \{0,1\}$ and a post selection criteria $h:\{0,1\}^n\rightarrow \{0,1\}$ let
    
    \[
    p=\Pr_{x\in \{0,1\}^n}[f(x)=1| h(x)=1]
    =\frac{\sum_{x\in \{ 0 ,1 \}^n}f(x)h(x)}{\sum_{x\in \{ 0 ,1 \}^n}h(x)}
    .
    \]
    
    Then for all $g\geq 1+ \frac{1}{poly(n)}$, there exists an $\mathsf{FBPP_{\parallel}^{NP^{f,h}}}$ machine that approximates $p$ to within a multiplicative factor of $g.$
\end{lemma}

The proof of lemma \ref{lma: postselected stockmeyer} is quite trivial once it is realised that Stockmeyer's approximate counting theorem can be done with only parallel oracle calls, which we capture in the next lemma.

\begin{lemma}[Approximate counting in parallel]\label{lma:better stockmeyer}
    Given a Boolean function $f:\{0,1\}^n\rightarrow \{0,1\}$, let
    \[
    p=\Pr_{x\in \{0,1\}^n}[f(x)=1]=\frac{1}{2^n} \sum_{x\in \{ 0 ,1 \}^n}f(x).
    \]
    Then for all $g\geq 1+ \frac{1}{poly(n)}$, there exists an $\mathsf{FBPP_{\parallel}^{NP^f}}$ machine that approximates p to within a multiplicative factor of $g.$
\end{lemma}

We will not provide all the steps of this lemma as it is a clear corollary of the version of the proof given by Valiant and Vazirani \cite{valiant1985np}\footnote{To see this, note that they can fix their random vectors at the start of their algorithm and check in parallel if 1, 2, up to $n$ vectors are sufficient to empty the set.}. This theorem is also a consequence of the results of O'Donnel and Say \cite{o2018weakness}, who show that approximate counting is in $\mathsf{PostBPP}$, and therefore in $\mathsf{BPP^{NP}_{\parallel}}$.

We then proceed with the proof of lemma \ref{lma: postselected stockmeyer}.

\begin{proof}[Proof of lemma \ref{lma: postselected stockmeyer}]
    Fix a target error, $g$, from the assumption of the theorem there exists $d$ such that $g>1+\frac{1}{n^d}$. By lemma \ref{lma:better stockmeyer} we can approximate $\sum_{x\in \{ 0 ,1 \}^n}f(x)h(x)$ to a multiplicative factor of ${g'}=1+\frac{1}{3n^d}$ with high probability. Similarly we can approximate $\sum_{x\in \{ 0 ,1 \}^n}h(x)$ to ${g'}$ as well. Dividing the first sum by the second gives us $p$ within a multiplicative factor of 
    \[
    {g'}^2 = 1+\frac{2}{3n^d}+\frac{1}{9n^{2d}} \leq 
    1+\frac{2}{3n^d}+\frac{1}{9n^{d}} = 1+\frac{7}{9}\frac{1}{n^d} < 1+\frac{1}{n^d}
    \]
with sufficiently high probability. Since each of these sums can be done in parallel and then divided for the final answer we can perform them in parallel. Given we only need to decrease the failure probability by a factor of 2 the final algorithm is in $\mathsf{FBPP^{NP}_{\parallel}}^f$.
\end{proof}

We will first apply the above logic for $\mathsf{BosonSampling}$ results.
Aaronson and Arkhipov use one post-selection step in lemma 42 in \cite{aaronson2011computational}, they then perform one approximate counting step on a quantity derived from this post-selection step to prove their main theorem. The structure of this proof fits the format of lemma \ref{lma: postselected stockmeyer}, therefore we can state a parallelised version of their main theorem.

\begin{theorem}[Aaronson and Arkhipov with parallel calls]
    Let $\mathcal{D}_A$ be the probability distribution sampled by a boson computer $A$. Suppose there exists a classical algorithm C that takes as input a description of A as well as an error bound $\varepsilon$, and that samples from a probability distribution $\mathcal{D}'_A$ such that $\norm{\mathcal{D}'_A - \mathcal{D}_A}\leq \varepsilon$ in poly$(|A|, 1/\varepsilon)$ time. Then the $|GPE|^2_{\pm}$ problem is solvable in $\mathsf{BPP^{NP}_{\parallel}}$.
\end{theorem}

The $|GPE|^2_{\pm}$ problem (defined in \cite{aaronson2011computational}) becomes $\mathsf{\# P}$ complete given two conjectures: The \textbf{permanent-of-Gaussians Conjecture} (PGC) and the \textbf{Permanent Anti-Concentration Conjecture} (PACC). These conjectures capture the belief that the permanent of Gaussian matrices does not concentrate close to zero and is hard to estimate, further justification of these conjectures is available in the original paper \cite{aaronson2011computational}. These are the \textit{additional assumptions} we referred to earlier.

\begin{theorem}
Assume PACC and PGC hold. If there exists a classical algorithm that can sample the distribution of a boson computer with additive error, then $\mathsf{P^{\# P} = ZPP^{NP}}$ and the polynomial hierachy collapse to the second level.
\end{theorem}

Next, we discuss how the same strategy applies to the $\mathsf{IQP}$ case. The proof provided by Bremner, Montanaro and Shepherd uses just Stockmeyer counting to prove their collapse to $\mathsf{BPP^{NP}}$ in their theorem 1. 
By lemma \ref{lma:better stockmeyer} all $\mathsf{NP}$ calls in their algorithm can be done in parallel and we can convert their theorem 1 to a $\mathsf{BPP^{NP}_{\parallel}}$ result. For $\mathsf{IQP}$ we do not need to conjecture the concentration of some hard-problem. The result only rest on the conjectured average case hardness of approximately computing either the partition function of the Ising model \textit{or} on a property of low-degree polynomials over finite fields (Conjectures 2 and 3 in \cite{bremner2016average}).

\begin{theorem}[Improved IQP hardness \cite{bremner2016average}]
Assume either above conjecture is true. If it is possible to classically sample from the output probability distribution of any $\mathsf{IQP}$ circuit in polynomial time, up to an error of 1/192 in $l_1$ norm, then $\mathsf{P^{\#P}}=\mathsf{ZPP^{NP}}$. Hence the Polynomial Hierarchy would collapse to its second level.
\end{theorem}

Finally, the same strategy can be applied to improve Morimae's result \cite{morimae2017hardness} on the hardness of the $\mathsf{DQC1}$ model for sampling \cite{fujii2015power} as it again depends on Stockmeyer counting. The necessary conjecture for Morimae's is an assumption directly about the average case hardness of the one clean qubit model.

\begin{theorem}[Improved one clean qubit hardness \cite{morimae2017hardness}]
    Assuming Morimae's conjecture, if there exists a classical algorithm which can output samples from any one clean qubit machine with at most 1/36 error in the $l_1$ norm then there is a $\mathsf{ZPP^{NP}}$ algorithm to solve any problem in $\mathsf{P^{\#P}}$. Hence the Polynomial Hierarchy would collapse to its second level.
\end{theorem}

It seems likely that other results will fit the structure we give here, although we provide only these three results.

We finish this section by noting that the above results are in $\mathsf{SampBQP}$, so classical impossibility on any of these tasks would imply $\mathsf{SampBQP\neq SampP}$, which would also imply a separation in the functional classes $\mathsf{FBQP \neq FBPP}$\cite{aaronson2014equivalence}.
\begin{theorem}
    If any of the sets of conjecture above hold and the polynomial hierarchy does not collapse to it's second level, then $\mathsf{FBQP\neq FBPP}$ and equivalently $\mathsf{SampBQP \neq SampP}$.
\end{theorem}

\section{Conclusion and Open Problems}
This paper has shown that if $\mathsf{P^{\#P}=BPP^{NP}_{\parallel}}$ then the polynomial hierarchy collapses to its second level. We have connected this result to approximate/exact counting and shown that it improves a number of results demonstrating the separation of quantum and classical computing.
The natural next research question is `how low can we go?'. Fujii et al \cite{fujii2015power} extended the result for the multiplicative error case to $\mathsf{AM \cap coAM}$, perhaps this could be achieved. We have not used the fact that the $\mathsf{ZPP^{NP}}$ algorithm likely only needs two rounds of oracle calls. This could be an avenue to collapsing the hierarchy further.

This hierarchy collapse strengthens claims of quantum supremacy but as these claims also rest on a number of additional assumptions (e.g. permanents-of-Gaussians conjectures), diminishing, removing or proving the other assumptions remain one of the central challenges of showing quantum supremacy and proving $\mathsf{SampBQP\neq SampP}$. Alternatively further work may reveal one of these assumptions to fall through, such as with XQUATH \cite{aharonov2023polynomial}.

Our results offer other, more direct, extensions. Of particular interest is whether our main theorem relativises, like previous supremacy results did \cite{aaronson2011computational}. Avoiding using the checkability of $\mathsf{\#P}$ may be key to proving relativisation as this is the only step of our proof that did not relativise.

Another open question is how Theorem \ref{thm: main} may be used for other non-quantum purposes, perhaps where approximate and exact counting are being compared.
Alternatively, a research line we did not pursue is extending theorem \ref{thm: main}. Extensions could be to other checkable rsr languages, perhaps $\mathsf{PSPACE}$ or $\mathsf{EXP}$ (although these may only be adaptively-rsr \cite{feigenbaum1990random}), or to showing the equality holds for other elements of the polynomial hierarchy ($\mathsf{P^{\# P}=BPP^{\Sigma_i}_{\parallel}\stackrel{?}{\implies} P^{\# P}=ZPP^{\Sigma_i}} $).

\section*{Acknowledgements}
SCM thanks Jan H. Kirchner for their insightful comments.
VD and SCM acknowledge the support by the project NEASQC funded from the European Union’s Horizon 2020 research and innovation programme (grant agreement No 951821). VD and SCM also acknowledge partial funding by an unrestricted gift from Google Quantum AI.
VD was supported by the European Union’s Horizon Europe program through the ERC CoG BeMAIQuantum (Grant No. 101124342) and the Dutch National Growth Fund (NGF), as part of the Quantum Delta NL program.
SA was supported by a Vannevar Bush Fellowship from the US Department of Defense, the Berkeley NSF-QLCI CIQC Center, and a Simons Investigator Award.
\bibliographystyle{plain}
\bibliography{biblo.bib}

\end{document}